\documentclass[11pt]{article}

\usepackage[utf8]{inputenc}

\usepackage{amsmath,amsfonts,amsthm,amssymb,color}
\usepackage[usenames,dvipsnames,svgnames,table]{xcolor}
\definecolor{darkgreen}{rgb}{0.0,0,0.9}
\usepackage{mathtools}
\usepackage{authblk}
\usepackage{fullpage}
\usepackage{parskip}
\usepackage{comment}
\usepackage{tikz}
\usepackage{bbm}
\usepackage{dsfont}
\usepackage[sc]{mathpazo}
\usepackage[basic]{complexity}
\usepackage{algorithm2e}
\usepackage[capitalize,nameinlink]{cleveref}
\usepackage{tcolorbox}
\usepackage[short]{optidef}
\usepackage{graphicx}

\newtcolorbox{wbox}
{
	colback  = white,
}

\SetKwInOut{Input}{Input}
\SetKwInOut{Output}{Output}
\SetKwFunction{Uncross}{\textsc{Uncross}}
\SetKwFunction{MergeUncross}{\textsc{MergeUncross}}
\SetKwFunction{PerfectMatching}{\textsc{PerfectMatching}}
\SetKwBlock{InParallel}{in parallel do}{end}
\SetKwFor{ParallelFor}{for}{in parallel do}{end}

\let\R\relax
\newcommand*{\R}{\mathbb{R}}

\newcommand*{\suppress}[1]{}

\newcommand*{\cT}{\mathcal{T}}

\makeatletter
\def\thm@space@setup{%
	\thm@preskip= 10pt
	\thm@postskip=\thm@preskip 
}
\makeatother

\makeatletter
\renewcommand{\paragraph}{%
	\@startsection{paragraph}{4}%
	{\z@}{5pt}{-1em}%
	{\normalfont\normalsize\bfseries}%
}
\makeatother

\newtheorem{theorem}{Theorem}
\newtheorem{lemma}[theorem]{Lemma}
\newtheorem{corollary}[theorem]{Corollary}

\theoremstyle{definition}
\newtheorem{definition}[theorem]{Definition}
\newtheorem{notation}[theorem]{Notation}
\newtheorem{remark}[theorem]{Remark}

\newtheorem{alg}[theorem]{Algorithm}

\newenvironment{fminipage}%
{\begin{Sbox}\begin{minipage}}%
		{\end{minipage}\end{Sbox}\fbox{\TheSbox}}

\newcommand{\CS}{\mbox{${\cal S}$}}

\newcommand\QQ{\boldsymbol{\mathit{Q}}}

\newcommand{\cuts}{\mbox{\rm cuts}}

\newcommand{\Lc}{\mathcal{L}}

\title{An Extension of the Birkhoff-von Neumann Theorem \\
to Non-Bipartite Graphs}

\author[1]{Vijay V.~Vazirani}

\affil[1]{University of California, Irvine}

\date{}

\begin{document}
	\maketitle

\begin{abstract}
	We prove that a fractional perfect matching in a non-bipartite graph can be written, in polynomial time, as a convex combination of perfect matchings. This extends the Birkhoff-von Neumann Theorem from bipartite to non-bipartite graphs.
	
	The algorithm of Birkhoff and von Neumann is greedy; it starts with the given fractional perfect matching and successively ``removes'' from it perfect matchings, with appropriate coefficients. This fails in non-bipartite graphs -- removing perfect matchings arbitrarily can lead to a graph that is non-empty but has no perfect matchings. Using odd cuts appropriately saves the day. 
\end{abstract}
    
\bigskip
\bigskip
\bigskip
\bigskip
\bigskip
\bigskip
\bigskip
\bigskip

\bigskip
\bigskip
\bigskip
\bigskip
\bigskip
\bigskip
\bigskip
\bigskip

\bigskip
\bigskip
\bigskip
\bigskip
\bigskip
\bigskip

\pagebreak
    
\section{Introduction}
\label{sec:intro}

A classic theorem due to Birkhoff \cite{Birkhoff1946tres} and von Neumann \cite{von1953certain}, which we will call the BvN Theorem, states that any doubly stochastic matrix can be efficiently expressed as a convex combination of permutation matrices. Equivalently, the task performed by the BvN Algorithm is: Given a fractional perfect matching, decompose it into a convex combination of perfect matchings in polynomial time. We give the analogous result for non-bipartite graphs.

Formally, the BvN Theorem says that a fractional perfect matching\footnote{In this formulation, we will view each (fractional) perfect matching as a vector in $\R_+^{n^2}$.} $x$ in a bipartite graph $G = (U, V, E)$ can be decomposed into a convex combination of perfect matchings, i.e., there are perfect matchings $x_1, \ldots , x_k$ in $G$ such that 
\[ x = a_1 \cdot x_1 + \ldots + a_k \cdot x_k , \]
where each $a_i \geq 0$ and $\sum_{i = 1}^k {a_i} = 1$. We will call each $a_i \cdot x_i$ a {\em term} of the decomposition; it turns out that $n^2 -2n + 2$ terms suffice, see \cite{LP.book}.

In this paper, we prove that a fractional perfect matching in a non-bipartite graph can be decomposed into a convex combination of perfect matchings in polynomial time. We place an upper bound of $m$ on the number of terms in our decomposition, where $m$ is the number of positive components in the vector specifying the fractional perfect matching.

\begin{figure}[htb]
\begin{center}
\includegraphics[height=2.8in,width=4.8in,angle=0]{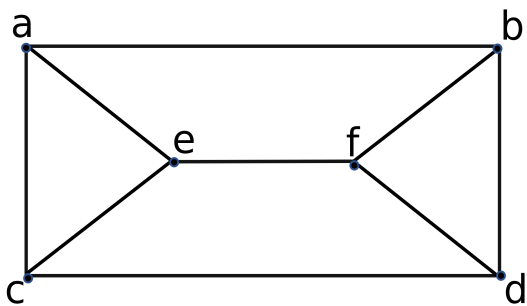}
\caption{Graph $G_1$, illustrating that the greedy approach fails in the non-bipartite case. \label{fig.counter-example}}
\end{center}
\end{figure}

In a sense, the BvN Theorem and our result are subsumed by the following very general fact. The Caratheodory theorem states that any point in a polytope $P$ can be written as convex combination of $N+1$ vertices of $P$, where $N$ is the dimension\footnote{The dimension of the perfect matching polytope of a matching covered graph is $m - n + 1 - k$, where $k$ is the number of bricks in the canonical decomposition of the graph; see \cite{LP.book} and Remark \ref{rem.brick}. A graph is {\em matching covered} if its every edge is in some perfect matching.} of $P$. Given a separation oracle for the polytope, there is even a polynomial time algorithm for computing this decomposition using the Ellipsoid algorithm \cite{GLS}. However, the resulting algorithm is far from practical, and moreover this fact gives no combinatorial insights into the individual problems to which it is applied. Consequently, the value of the BvN Theorem was not diminished by the discovery of  this fact and we believe the same should hold for our result. However, this fact is valuable in its own right -- for the high-level picture it provides. 

 The BvN Algorithm is greedy: it successively finds perfect matchings $x_1, \ldots , x_k$ in $G$, as well as the associated coefficients $a_1, \ldots , a_k$. We first show that a greedy approach will fail in a non-bipartite graph. Consider the graph $G = (V, E)$ given in Figure \ref{fig.counter-example} which we will call {\em graph $G_1$}. Let $x$ be the vector on $E$ having $1/3$ in each component; one can check that $x$ is a fractional perfect matching, see Section \ref{sec.preliminaries} for a formal definition. Suppose the first perfect matching found by the greedy approach is $M = \{(a, b), (c, d), (e, f)\}$. On subtracting $(1/3) M$ from $x$, we are left with a vector whose support has no perfect matchings; the support consists of the two ``triangles'' $(a, c, e)$ and $(b, f, d)$. It turns out that $x$ has a unique decomposition consisting of three terms each of which includes a distinct edge from $M$. Hence $M$ cannot be included to a non-zero extent in a decomposition of $x$.
 
 Next, consider the Petersen graph given in Figure \ref{fig.Petersen} which we will call {\em graph $G_2$}. Again, $1/3$ on each edge is a fractional perfect matching. Suppose the first perfect matching found is the five ``spokes'', say $M = \{(a, f), (b, g), (c, h), (d, i), (e, j)\}$. On subtracting $(1/3) M$ from $x$, we are left with a vector whose support consists of two disconnected 5-cycles and hence has no perfect matchings. Once again the graph has a unique decomposition, but this time it includes $M$ to the extent of $1/6$. 
 
Let us summarize the two issues raised by these examples. First, the perfect matching picked cannot be arbitrary. Second, it cannot be subtracted maximally from $x$; the algorithm needs to figure out the ``right'' coefficient of this term. Our algorithm gets around both these difficulties by dealing carefully with odd subsets of vertices. In contrast, in the bipartite case, picking an arbitrary perfect matching from the support of $x$ and subtracting it maximally from $x$ works; indeed that is the BvN Algorithm.

\begin{figure}[htb]
\begin{center}
\includegraphics[height=2.7in,width=4.8in,angle=0]{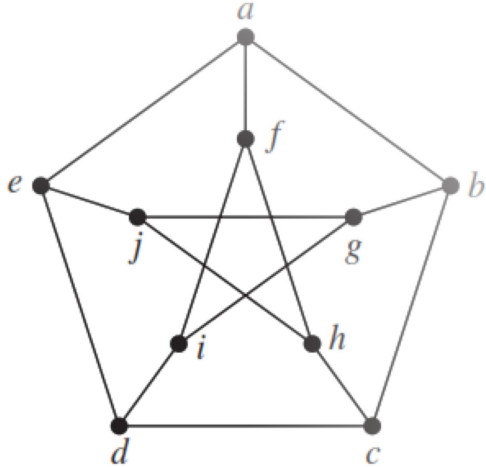}
\caption{Graph $G_2$, the Petersen graph. \label{fig.Petersen}}
\end{center}
\end{figure}

The BvN Theorem and Algorithm have numerous applications. Of special interest to us is application to one-sided matching markets under which $n$ agents need to be matched to $n$ indivisible goods, one per agent. The well-known Hylland-Zeckhauser \cite{hylland} scheme first renders each good divisible by viewing it as one unit of probability shares. It then allocates a total of one unit of probability shares to each agent. Clearly, this allocation is a fractional perfect matching over agents and goods. Finally, an integral perfect matching is found via randomized rounding, using the BvN Algorithm. Matching markets whose underlying graphs are non-bipartite also have important applications, e.g., to the roommates problem \cite{GaleS} and kidney exchange \cite{Roth2005pairwise}. The question of extending the BvN Theorem to non-bipartite graphs arose in that context.

\section{Preliminaries}
\label{sec.preliminaries}

\subsection{The Bipartite Case}
\label{sec.bipartite}

LP (\ref{eq.bipartite}) is a linear programming relaxation of the perfect matching problem in a bipartite graph $G = (U, V, E)$, where $|U| = |V| = n$. The polytope defined by constraints \ref{const.1} and \ref{const.2} is the convex hull of all perfect matchings in $G$, hence any point $x$ in this polytope is a convex combination of perfect matchings. We will call $x$ a {\em fractional perfect matching} in $G$. 

	\begin{maxi!}
		{} {\sum_{e \in E}  {x_{e}}}
			{\label{eq.bipartite}}
		{}
		\addConstraint{\sum_{e \in \delta(v)} {x_{e}}}{=1 \quad \label{const.1}}{\forall v \in (U \cup V)}
		\addConstraint{x_{e}}{\geq 0 \quad \label{const.2}}{\forall e \in E}
	\end{maxi!}
	
Throughout we will assume that the fractional perfect matching specified in the given instance has rational entries, i.e., $x \in \QQ_+^{n^2}$. For $\alpha \in \QQ_+$, with $0 < \alpha \leq 1$, define an {\em $\alpha$-fractional perfect matching in bipartite graph $G$} to be a feasible solution to the modification of LP (\ref{eq.bipartite}) obtained by replacing 1 in the right hand side of constraint \ref{const.1} by $\alpha$. Observe that $x$ is an $\alpha$-fractional perfect matching if and only if ${1 \over \alpha} x$ is a fractional perfect matching in $G$. 

We can now formally state the BvN Algorithm. The algorithm is iterative and we will call each of its iterations a {\em phase}. Given a fractional perfect matching $x$, clearly the support of $x$ contains a perfect matching. Find one such matching, say $M$. Let $\beta = \min_{e \in M} \{x_e\}$. Then $\beta \cdot M$ will be the first term\footnote{By a slight abuse of notation, $M$ will represent the edges of the perfect matching as well as the corresponding point (vector) in the polytope.} of the decomposition.

Next, subtract $\beta$ from all components of $x$ which correspond to edges of $M$, to obtain $x'$. Henceforth, we will denote this operation by $x' = (x - \beta \cdot M)$. Clearly in going from $x$ to $x'$, at least one edge of $M$ has been set to zero and will no longer be in the support of $x'$. It is easy to see that $x'$ is a $(1 - \beta)$-fractional perfect matching in $G$. As remarked above, ${1 \over {1 - \beta}} x'$ is a fractional perfect matching in $G$. Therefore its support contains a perfect matching and hence we can continue the process until $x'$ turns out to be the zero vector. Since at least one edge is dropped in each phase, the procedure takes polynomial time.  

\subsection{The Non-Bipartite Case}
\label{sec.non-bipartite}

LP (\ref{eq.non-bipartite}) is a linear programming relaxation of the perfect matching problem in a non-bipartite graph $G = (V, E)$, where $|V| = n$ is even. Again, the polytope defined by constraints \ref{const.1non} and \ref{const.2non} is the convex hull of all perfect matchings in $G$ and any point $x$ in this polytope is a {\em fractional perfect matching} in $G$. Throughout this paper, by an {\em odd set in $G$} we mean a set $S \subset V, \ \mbox{with} \ |S| \ \mbox{odd and} \ |S| \geq 3$.

	\begin{maxi!}
		{} {\sum_{e \in E}  {x_{e}}}
			{\label{eq.non-bipartite}}
		{}
		\addConstraint{\sum_{e \in \delta(v)} {x_{e}}}{= 1 \quad \label{const.1non}}{\forall v \in V}
		\addConstraint{\sum_{e \in \delta(S)} {x_{e}}}{\geq 1 \quad \label{const.2non}}{\forall \  \mbox{odd set} \ S}
		\addConstraint{x_{e}}{\geq 0}{\forall e \in E}
	\end{maxi!}
	
For $S \subset V$, the {\em cut of $S$}, denoted $\delta(S)$, is the set of edges having exactly one endpoint in $S$, i.e., $\delta(S) = \{(u, v) \in E \ | \ \ |\{u, v \} \cap S| = 1\}$. Odd sets $S$ for which constraint \ref{const.2non} is tight are called {\em tight odd sets}; clearly, every perfect matching crosses a tight odd set exactly once. As is common practice, we will use ``odd set'' and ``odd cut'' synonymously. 


For $\alpha \in \QQ_+$, with $0 < \alpha \leq 1$, define an {\em $\alpha$-fractional perfect matching in $G$} to be a feasible solution to the modification of LP (\ref{eq.non-bipartite}) obtained by replacing 1 in the right hand side of constraints \ref{const.1non} and \ref{const.2non} by $\alpha$. Since this LP is central to our algorithm, we give it explicitly as LP (\ref{eq.alpha}). As in the bipartite case, $x$ is an $\alpha$-fractional perfect matching if and only if ${1 \over \alpha} x$ is a fractional perfect matching in $G$.

	\begin{maxi!}
		{} {\sum_{e \in E}  {x_{e}}}
			{\label{eq.alpha}}
		{}
		\addConstraint{\sum_{e \in \delta(v)} {x_{e}}}{= \alpha \quad \label{const.1alpha}}{\forall v \in V}
		\addConstraint{\sum_{e \in \delta(S)} {x_{e}}}{\geq \alpha \quad \label{const.2alpha}}{\forall \  \mbox{odd set} \ S}
		\addConstraint{x_{e}}{\geq 0}{\forall e \in E}
	\end{maxi!}

For any vector $x \in \R_+^m$, let $E_x$ denote the set of edges which form the support of $x$ and let $G_x = (V, E_x)$ denote the {\em support graph of $x$}. In $G_x$, the cut of $S$ will be denoted by $\delta_x (S)$, i.e., $\delta_x(S) = \{(u, v) \in E_x \ | \ |\{u, v \} \cap S| = 1\}$ and $x(\delta_x(S))$ will denote $\sum_{e \in \delta_x (S)} {x_e}$. By ``find a minimum odd cut in $(G_x, x)$'' we mean find an odd set $S$, with $|S| \geq 3$, which minimizes $x(\delta_x (S))$, i.e., a minimum odd cut in graph $G_x$ with edge weights given by $x$. 

The next definition has been highlighted because it is central to this paper.

\begin{definition}
	\label{def.tight-alpha}
	Let $x$ be an $\alpha$-fractional perfect matching in $G$, i.e., it is a feasible solution to  LP (\ref{eq.alpha}). For an odd set $S \subset V$, we will say that $S$ is a {\em tight odd $\alpha$-cut} if $x(\delta_x(S)) = 1$. Observe that every tight odd cut in $G$ is also a tight odd $\alpha$-cut in $G_x$, since every perfect matching crosses a tight odd cut exactly once. 
\end{definition}

As in the bipartite case, the vector specified in the given instance has rational entries, i.e, $x \in \QQ_+^{{n(n-1)} \over 2}$.

\begin{notation}
	\label{not.edges}
	Corresponding to the input vector $x$, we will denote the number of edges in its support graph $G_x$ by $m$. We will denote the largest denominator occurring in components of $x$ by $d$.
\end{notation}

	\section{High Level Idea of the Algorithm}
\label{sec.high}

Given a fractional perfect matching $x$, our algorithm needs to execute the following two tasks in the first phase:
\begin{enumerate}
	\item Find a perfect matching $M$ in $G_x$ that belongs to some decomposition of $x$, i.e., such that $(x - b \cdot M)$ is a $(1 - b)$-fractional perfect matching for some $b > 0$. 
	\item Find the largest number, say $\beta$, such that $(x - \beta \cdot M)$ is a $(1 - \beta)$-fractional perfect matching. 
\end{enumerate}  

 Then the first term of the decomposition will be $\beta \cdot M$ and the algorithm will proceed to decompose $(x - \beta \cdot M)$. For the first task, we will find a maximal laminar family of tight odd $\alpha$-cuts in $G_x$, say $\Lc$, and use the following two lemmas.  

\begin{lemma}
	\label{lem.once}
	Let $b \cdot M$ be a term in some decomposition of $x$, where $M$ is a perfect matching in $G_x$. Then, $M$ crosses every cut $C \in \Lc$ exactly once. 
\end{lemma}

\begin{proof}
	Clearly, $M$ must cross every $C \in \Lc$ at least once. Let $x' = (x - b \cdot M)$. Since $b \cdot M$ is a term in a valid decomposition of $x$, $x'$ must be a $(1 - b)$-fractional perfect matching.  If $M$ crosses some cut $C$ more than once, then the corresponding constraint in \ref{const.2non} will be violated with respect to $x'$, contradicting the previous assertion. The lemma follows.	
	\end{proof}

For each $e \in E_x$, define $ \cuts(e) = \{ C \in \Lc \ | \ e \in \delta(C) \}$,
i.e., the set of cuts in $\Lc$ which $e$ crosses, and define the {\em weight of edge $e$} to be the cardinality of this set, i.e., $w(e) = |\cuts(e)|$. Next, find a minimum weight perfect matching in the graph $G_x$ under weight function $w$.   

\begin{lemma}
	\label{lem.min-wt}
	Let $M$ be a minimum weight perfect matching in graph $G_x$ under weight function $w$. Then the weight of this matching is $|\Lc|$ and it crosses each tight odd $\alpha$-cut exactly once.
\end{lemma}

\begin{proof}
By Lemma \ref{lem.once}, there is a perfect matching which crosses each cut in $\Lc$ exactly once and therefore has weight $|\Lc|$. Further, every perfect matching must cross each cut in $\Lc$ at least once and therefore must have weight at least $|\Lc|$. Hence the weight of a minimum weight perfect matching is $|\Lc|$. 

Since $M$ be a minimum weight perfect matching, its weight is $|\Lc|$ and it crosses each cut $C \in \Lc$ exactly once. Clearly, $M$ must cross each tight odd $\alpha$-cut at least once. Suppose $M$ crosses some tight odd $\alpha$-cut, say $C'$, more than once. Since $\Lc$ is a maximal laminar family of tight odd $\alpha$-cuts, via standard uncrossing arguments one can show that there must be a cut $C \in \Lc$ such that $M$ crosses $C$ more than once, leading to a contradiction.
\end{proof}

Let us consider graphs $G_1$ and $G_2$ again in light of these facts. It turns out that neither of these graphs has a tight odd cut, see Remark \ref{rem.brick}. $G_1$ has two tight odd $\alpha$-cuts, namely those defined by $\{a, c, e\}$ and $\{b, f, d\}$ and $G_2$ has no tight odd $\alpha$-cuts. Therefore, in $G_1$, the three edges $\{(a, b), (c, d), (e, f)\}$ have weight 2 and the rest have weight zero. Hence each of the three perfect matchings, which picks exactly one of these edges, is a minimum weight perfect matching. In $G_2$, all edges have weight zero and every perfect matching is minimum weight. 

\begin{remark}
	\label{rem.brick}
	Both these graphs $G_1$ and $G_2$ are {\em bricks}: they are 3-connected and bicritical (i.e., the removal of any pair of vertices leaves a graph having a perfect matching). It turns out that bricks have no tight odd cuts; see \cite{LP.book} for details. 
\end{remark}

In the first phase, we will pick a minimum weight perfect matching, say $M$. Let $\beta = \min_{e \in M} \{x(e)\}$. For $0 < b \leq \beta$, let $x' = (x - b \cdot M)$. The next fact follows from Lemma \ref{lem.min-wt}. 

\begin{lemma}
	\label{lem.partial}
	For $\alpha = (1 - b)$, consider LP (\ref{eq.alpha}) with respect to $b$ and $x'$ defined above. Then the following hold:
	\begin{enumerate}
		\item Constraints \ref{const.1alpha} are all satisfied.
		\item Constraints \ref{const.2alpha} are satisfied for all tight odd $\alpha$-cuts.
	\end{enumerate}
\end{lemma}

We next determine if for $b = \beta$, $x'$ is a $(1 - \beta)$-fractional perfect matching, i.e. if constraints \ref{const.2alpha} of LP (\ref{eq.alpha}) are satisfied for all odd cuts which are not tight odd $\alpha$-cuts. If so, $\beta \cdot M$ will be the first term of the decomposition and we will proceed to the next phase to  decompose $x'$. If not, there is a odd set $S$ which violates constraint \ref{const.2alpha}. This happens because $|\delta(S) \cap M| > 1$, i.e., $M$ picks more than one edge from this cut. Then, in going from $x$ to $x'$, such a cut loses $|\delta(S) \cap M| \cdot \beta$ and therefore drops below $1 - \beta$. Since $S$ is not a tight odd $\alpha$-cut, $x(\delta_x (S)) > 1$. 

These two cases are illustrated by $G_1$ and $G_2$, respectively. In $G_1$, for any of the three matchings picked, say $N$, $\beta = 1/3$ and $x' = x - \beta \cdot N$ is a $(1 - \beta)$-fractional perfect matching. In $G_2$, the proposed matching, say $M$, picks five edges from the cut of $S = \{ f, g, h, i, j \}$ and $\min_{e \in M} \{x_e\} = 1/3$. However, observe that $x' = (x - {1 \over 3} \cdot M)$ is not a ${2 \over 3}$-fractional perfect matching, since $S$ is violated. 

We will deal with the latter problem by finding an appropriate smaller coefficient for $M$. To find this coefficient, say $b$, we solve:
\[ 5 \cdot \left( {1 \over 3} - b \right) = 1 - b , \]
giving $b = 1/6$. Clearly, $S$ is not violated in $x' =  (x - {1 \over 6} \cdot M)$. Furthermore, since no other odd cut is violated, $x'$ is a ${5 \over 6}$-fractional perfect matching. Therefore, ${1 \over 6} \cdot M$ will be the first term of the decomposition and we will proceed to the next phase to decompose $x'$.

However, in general, another odd cut may be violated with respect to $x'$. The recourse is to find the most stringent cut, which will yield an appropriate smaller value of $b$. This is defined below by $\gamma$. In this definition, $\CS = \{\mbox{odd set} \ S \ \ s.t. \ \ |\delta_x(S) \cap M| > 1\}$.

\begin{equation}
	\label{eq.gamma-one}
	\gamma \ = \ \min_{S \in \CS} \ \ \left( {{x(\delta_x(S)) - 1} \over {|\delta_x(S) \cap M| - 1}} \right) .
\end{equation}

For $S \in \CS$, $|\delta_x(S) \cap M| > 1$. Therefore $S$ is not a tight odd $\alpha$-cut and $x(\delta_x(S)) > 1$. Hence $\gamma > 0$.

Finally we note that the procedure for decomposing $x'$ in the second phase is identical to that of   the first phase. This follows from our observation that $x'$ is an $\alpha$-fractional perfect matching if and only if ${1 \over \alpha} x'$ is a fractional perfect matching in $G$. As a consequence of this statement, Lemmas \ref{lem.once}, \ref{lem.min-wt} and \ref{lem.partial}  hold with respect to $x'$ as well.

\section{The Algorithm}
\label{sec.alg}

In this section, we will provide details about Algorithm \ref{alg.main}, which computes the decomposition of fractional perfect matching $x$. The input to the algorithm is $x$ and $\Lc$, the latter being a a maximal laminar family of tight odd $\alpha$-cuts in $G_x$. An efficient algorithm for finding $\Lc$ is given in Section \ref{sec.find-L}. Each iteration of Algorithm \ref{alg.main} is called a {\em phase}; it computes and outputs one term of the decomposition.

In order to keep notation simple, in Section \ref{sec.high} we described only the first phase of our algorithm; we now describe an arbitrary phase. It starts with an $\alpha$-fractional perfect matching, which we call $y$, for some $\alpha \in (0, 1]$. Steps 2(a) through 2(d) are self-explanatory. 

Step 2(e) and 2(f) are trying to determine if there is an odd cut in $(G_{y'}, y')$ having capacity $< (\alpha - \beta)$. The cut of each singleton is exactly $(\alpha - \beta)$. Therefore, in contrast to the algorithms given in Section \ref{sec.find-L}, the singleton cuts do not interfere with the main goal of Steps 2(e) and 2(f), and the algorithm of Padberg and Rao \cite{Padberg1982odd} suffices for Step 2(e) for finding a minimum odd cut in $(G_{y'}, y')$. If  the condition in Step 2(f) is true, $\beta \cdot M$ is a valid term; it is output and the algorithm moves to the next phase.

\bigskip

\setcounter{figure}{1} 

\begin{figure}

	\begin{wbox}
		\begin{alg}
		\label{alg.main}
		{\bf The Main Algorithm. \ Input = $(x, \ \Lc)$}\\
		\\
		\begin{enumerate}
			\item {\bf Initialization:} 
			 \begin{enumerate}
			 	\item $\alpha \leftarrow 1$.
			 	\item $y \leftarrow x$.
			     $\quad$
			 \end{enumerate}
					     $\quad$
					    
			\item {\bf While} $y \neq 0$ {\bf do:}
			$\quad$
			\begin{enumerate}
				$\quad$
			\item $\forall e \in E_y$: \ $w(e) \leftarrow |\{S \in \Lc \ | \ e \in \delta_y (S) \}|$.
			
						     $\quad$

			\item $M \leftarrow$  a minimum weight perfect matching in $G_y$ wrt. weights $w$.

				$\quad$
				
			\item  $\beta \leftarrow  \min_{e \in M} \{y_e\}$.  
			
						 $\quad$

			\item $y' \leftarrow (y - \beta \cdot M)$.
			
			$\quad$
			
			\item  $S \leftarrow$  a minimum odd cut in $(G_{y'}, y')$.   
			
			$\quad$
			
		\item	{\bf If} $y'(\delta_{y'} (S)) \geq (\alpha - \beta)$ {\bf then do:}
		
		$\quad$
		
			\begin{enumerate}
				\item {\bf Output} $\beta \cdot M$.
				\item $\alpha \leftarrow (\alpha - \beta)$. 
				\item $y \leftarrow y'$.
			\end{enumerate}
			
						     $\quad$

			\item {\bf Else do:}
			\begin{enumerate}
						 $\quad$

			\item Via a binary search in $[0, \beta]$, find $\gamma$ as defined in (\ref{eq.gamma-second}).
			\item {\bf Output} $\gamma \cdot M$.
			 \item $y \leftarrow (y - \gamma \cdot M)$.  
			 \item $\alpha \leftarrow (\alpha - \gamma)$. 	
			 \item $\Lc \leftarrow \ \mbox{update}(\Lc)$.		 
			 $\quad$
			\end{enumerate}

			\end{enumerate}

	\item {\bf end}.
		\end{enumerate} 
		\bigskip
		\end{alg}
	\end{wbox}
\end{figure} 

However, if the condition is false, then $y'$ is not an $(\alpha - \beta)$-fractional perfect matching. Then, as mentioned above, we need to find $\gamma$. For the first phase, $\gamma$ was defined in (\ref{eq.gamma-one}); for an arbitrary phase, it is defined below. The set $\CS$ was defined with (\ref{eq.gamma-one}) and again it is easy to see that $\gamma > 0$. 

\begin{equation}
	\label{eq.gamma-first}
	\gamma \ = \ \min_{S \in \CS} \ \ \left( {{x(\delta_x(S)) - \alpha} \over {|\delta_x(S) \cap M| - 1}} \right) .
\end{equation}

Let $S$ be the set which achieves the minimum in (\ref{eq.gamma-first}). Then $S$ is a newly created tight odd $\alpha$-cut with respect to $(y - \gamma \cdot M)$. 

One way of finding $\gamma$ is given in Procedure Find-Gamma, presented as Algorithm \ref{alg.one}. It starts with $\gamma = \beta$ and iteratively decreases $\gamma$ if some odd set is violated. This procedure does terminate with $\gamma$ but we were unable to place a polynomial upper bound on the number of iterations executed. On the other hand, via this procedure, we placed an upper bound on the growth of the bit-complexity of numbers. 

Below we give an equivalent definition of $\gamma$; we will use this to compute $\gamma$ via a binary search. The upper bound on bit-complexity will be used to show that binary search takes polynomial time. 

\begin{lemma}
	\label{lem.gamma-eq}
	The following is an equivalent definition of $\gamma$, which was defined in (\ref{eq.gamma-first}).
	\begin{equation}
\label{eq.gamma-second}
	\gamma \ = \ \arg \max_{b} \ \{(y - b \cdot M) \ \mbox{is a } (\alpha - b) \mbox{-fractional perfect matching} . \} 
	\end{equation}
\end{lemma} 

\begin{proof}
	Clearly, $\gamma$ as defined in (\ref{eq.gamma-first}) satisfies: for every odd set $T$, \ $x'(\delta_{x'} (T)) \geq (\alpha - \gamma)$, where $x' = (x - \gamma \cdot M)$. Therefore $x'$ is a $(\alpha - \gamma)$-fractional perfect matching. Furthermore, by (\ref{eq.gamma-first}), there is an odd set $S$ such that $x'(\delta_{x'} (S)) = (\alpha - \gamma)$. Note that $S$ is a newly created tight odd $\alpha$-cut, as stated above. Therefore $\gamma$ is the smallest smallest value of $b$ for which the condition in (\ref{eq.gamma-second}) holds.
\end{proof}

Next we describe the binary search conducted in Step 2(f) to we compute $\gamma$. Let $b$ denote the parameter on which we conduct the search in the range $[0, \beta]$. For each value of $b$, we will use the Padberg-Rao algorithm to find a minimum odd cut in $(G_{y'} , y')$, where $y' = (y - b \cdot M)$. If the capacity of this cut is $< (\alpha - b)$, then $b$ needs to be reduced further. 

As stated right after (\ref{eq.gamma-first}), with respect to $(y - \gamma \cdot M)$, a new tight odd $\alpha$-cut of value $(\alpha - \gamma)$ gets created. Although this cut may cross some cuts in $\Lc$, there is a new cut that can be added to $\Lc$ maintaining laminarity. The operation of update$(\Lc)$ in Step 2(g)(v), adds a maximal set of new odd tight $\alpha$-cuts to $\Lc$ while still maintaining laminarity. The procedure given in Section \ref{sec.find-L} shows how to do this.

\subsection{An Algorithm for finding $\Lc$}
\label{sec.find-L}

Given an $\alpha$-fractional perfect matching, $x$, we first show how to check, in polynomial time, if $(G_x, x)$ has a tight odd $\alpha$-cut, and if so, to find one. We then use this as a subroutine to find a maximal laminar family of such cuts. 

For the first task, the main difficulty is caused by singleton cuts, i.e., those containing one vertex. Each such cut has a capacity of $\alpha$, and running the Padberg-Rao algorithm in a straightforward manner on $(G_x, x)$ may yield a singleton cut and not a tight odd $\alpha$-cut.

 The following idea helps get around this problem. Consider a tight odd $\alpha$-cut, $S$, and let $e$ and $f$ be two edges in $\delta_x (S)$ such that $e$ and $f$ form a matching. Let $e = (a, b)$ and $f = (c, d)$. Obtain $y$ from $x$ by subtracting $\epsilon$ from $x_e$ and $x_f$, where $\epsilon = {1 \over {Dn}}$, where $D$ is a bound on the largest denominator encountered by Algorithm \ref{alg.main}. Such a bound is established in Lemma \ref{lem.largest}. 
 
 With respect to $y$, the smallest singleton cuts are given by the four vertices $a, b, c$ and $d$, and each has a capacity of $\alpha - \epsilon$. On the other hand, the capacity of the cut $\delta_y (S)$ is $\alpha - 2 \epsilon$. There may be other tight odd $\alpha$-cuts with respect to $x$ which have capacity $\alpha - 2 \epsilon$ with respect to $y$, but no odd cut has a lower capacity. Therefore, the Padberg-Rao algorithm is forced to find a tight odd $\alpha$-cut.
 
 {\bf Algorithm for finding a tight odd $\alpha$-cut:} For each pair of edges $e$ and $f$ which form a matching in $G_x$, reduce their values by $\epsilon$, which is defined above, to obtain $y$. Using the Padberg-Rao algorithm, find a minimum odd cut with respect to $y$. Finally, among all odd cuts found in this manner, output the smallest capacity cut.

{\bf Algorithm for finding a maximal laminar family of tight odd $\alpha$-cuts:} We initialize $\Lc$ with the odd cut found by the algorithm stated above and iterate the following steps.  

Let $\Lc$ be the laminar family of cuts found so far. Let us represent the cuts in $\Lc$, together with the entire set $V$, i.e., $\{V\} \cup \Lc$, as a tree $\cT$ with a node for each set and with $S'$ a child of $S$ if $S' \subset S$ and there is no set $S''$ in this collection such that $S' \subset S'' \subset S$. 

For each $S \in \Lc$, we want to know if there is a new tight odd $\alpha$-cut, say $T \subset S$ where $T \not \subset S'$ for $S' \in \Lc$ and $S' \subset S$ and such that $\Lc \cup \{T\}$ is laminar. If so, $T$ will be added to $\Lc$ and be made a child of $S$ per the above-stated rules for constructing $\cT$. This is accomplished by shrinking $V - S$ to one vertex and for each $S'$ that is a child of $S$, we shrink $S'$ to a distinct vertex. In this graph, we run the above-stated algorithm. If we find a new odd cut, $T$, we make it a child of $S$ in the tree. Repeating this for all nodes of the tree, including the new nodes created, we obtain a maximal laminar family of tight odd $\alpha$-cuts. 

\begin{lemma}
	\label{lem.time-find-L}
	The algorithm for finding a maximal laminar family of tight odd $\alpha$-cuts requires $O(n m^2)$ calls to the Padberg-Rao algorithm.
\end{lemma}

\begin{proof}
	Our algorithm for finding one tight odd $\alpha$-cut makes a call to the Padberg-Rao algorithm for each choice of the two edges $e$ and $f$, and therefore requires $O(m^2)$ calls. This is repeated for each node of tree $\cT$. Since $\cT$ has at most $n$ nodes, the lemma follows. 
\end{proof}

The operation of update$(\Lc)$ is performed in exactly the same way, i.e., we start with $\Lc$, construct its tree $\cT$ and find all new odd cuts that can be added to $\Lc$ and as nodes in $\cT$. Clearly, the time bound given in Lemma \ref{lem.time-find-L} applies to this procedure as well.

\subsection{Procedure Find-Gamma}
\label{sec.subroutine}

 A phase in Algorithm \ref{alg.main} can end in either Step 2(f) or 2(g). The growth of denominators of components of vector $y$ is different for these two cases; the first is straightforward to quantify and the second is done via Procedure Find-Gamma. This procedure is not called from Algorithm \ref{alg.main}. However, for the purpose of quantifying growth in the second case, we will call it in Step 2(g)(i) right after binary search has been executed. 
 
 The following input is provided to the procedure: $(\alpha, y, M, \beta, S)$, where $y$ is the $\alpha$-fractional perfect matching which Algorithm \ref{alg.main} was attempting to decompose in the current phase, and $M, \beta, S$ were computed in Steps 2(b), 2(c) and 2(e), respectively. 
 
 This procedure needs to find a minimum odd cut in Steps 1(e) and 2(d). As in Step 2(e) of Algorithm \ref{alg.main}, this is done using the algorithm of Padberg and Rao \cite{Padberg1982odd}.

\begin{lemma}
	\label{lem.gamma}
	Procedure Find-Gamma terminates and its output is precisely $\gamma$, as defined in (\ref{eq.gamma-first}). 
\end{lemma}

\begin{proof}
Consider the cuts computed in Step 2(d) in two successive iterations and say they are $S$ and $S'$, respectively. Assume that $|\delta_y(S) \cap M| = k$ and $|\delta_y(S') \cap M| = k'$. Let $b$ be the value of $\gamma$ in the iteration in which $S$ was found to be the minimum odd cut. Since in this iteration $S'$ was not picked as the minimum odd cut, 
	\begin{equation}
		\label{eq.subroutine}
			 \delta_y (S') - k' \cdot b \geq \delta_y (S) - k \cdot b  = \alpha - b.
	\end{equation}
	Let $c$ be the value of $\gamma$ in the next iteration, in which $S'$ was found to be the minimum odd cut. Therefore, $\delta(S') - k' \cdot c = \alpha - c$. Hence, $c < b$. This shows that the value of $\gamma$ computed in Step 2(b) is monotonically decreasing. Since the total number of cuts in $G_y$ is finite, the procedure must terminate.

Let $T'$ be the odd cut which makes the condition in the While-loop, in Step 2 of the procedure turn out to be true for the last time and let $T''$ be the odd cut found subsequently in Step 2(d); observe that $T'$ may equal $T''$. Let $k = |\delta_y (T') \cap M|$, as computed in Step 2(a). Then $\gamma$, as computed in Step 2(b), satisfies
\[ \delta_y (T') - k \gamma \ = \ \alpha - \gamma .\]
Therefore, $y'(\delta_{y'} (T')) = \alpha - \gamma$. The next time around, the condition in the While-loop, in Step 2, turns out to be false. Therefore, for all odd cuts $S$, $y'(\delta_{y'} (S)) \geq \alpha - \gamma$. Hence $\gamma$ is precisely as defined in (\ref{eq.gamma-first}).
\end{proof}

\bigskip

\setcounter{figure}{1} 

\begin{figure}

	\begin{wbox}
		\begin{alg}
		\label{alg.one}
		{\bf Procedure Find-Gamma.	\ Input = $(\alpha, y, M, \beta, S)$}
		\\
		\begin{enumerate}
			\item {\bf Initialization:} 
				 \begin{enumerate}
			 	\item $T \leftarrow S$.
			 	\item $k \leftarrow  |\delta_y (T) \cap M|$.
			 	\item $\gamma \leftarrow \left({{y(\delta_y (T)) - 1} \over {k-1}} \right)$. 
			 	\item $y' \leftarrow (y - \gamma \cdot M)$.
			 	\item $T \leftarrow$  a minimum odd cut in $(G_{y'}, y')$. 
			     $\quad$
			 \end{enumerate}
					     $\quad$

			\item  {\bf While} $y'(\delta_{y'} (T)) < \alpha - \gamma$ {\bf do:} 
			
						     $\quad$
						     
				\begin{enumerate}
			 	\item $k \leftarrow  |\delta_y (T) \cap M|$.
			 	\item $\gamma \leftarrow \left({{y(\delta_y (T)) - \alpha} \over {k-1}} \right)$. 
			 	\item $y' \leftarrow (y - \gamma \cdot M)$.
			 	\item $T \leftarrow$  a minimum odd cut in $(G_{y'}, y')$. 
			     $\quad$
			 \end{enumerate}
			 
			 $\quad$
			 
			 \item {\bf end.}

			$\quad$
			
			\item  {\bf Output} $(\gamma)$ 
			
						     $\quad$
						     
		\end{enumerate} 
		\bigskip
		\end{alg}
	\end{wbox}
\end{figure} 

\begin{lemma}
	\label{lem.denominators}
	Assume that every denominator occurring in $\alpha$ and $y$, which were input to Find-Gamma, is $b$. Then the denominator of $\gamma$ output by the procedure is bounded by $b \cdot {n}$.
\end{lemma}

\begin{proof}
Consider the very last execution of Step 2(b); this computes $\gamma$ which is output in Step 4. If $T$ is the odd set at this stage and $k = |\delta_y (T) \cap M|$, then 
\[ \gamma =  \left({{y(\delta_y (T)) - \alpha} \over {k-1}} \right) .\] 

Clearly the denominator of $(y(\delta_y (T)) - \alpha)$ is $b$. Since $k \leq |M| = {n \over 2}$, the lemma follows. 
\end{proof}

\subsection{Upper Bounding Running Time and the Number of Terms}
\label{sec.time}

\begin{lemma}
	\label{lem.phases}
	The number of phases executed by Algorithm \ref{alg.main} is at most $m$, where $m$ is defined in Notation \ref{not.edges}. 
\end{lemma}

\begin{proof}
Let us say that a phase of Algorithm \ref{alg.main} is of {\em Type 1} if the condition in Step 2(f) is true and of {\em Type 2} otherwise. As in the BvN Algorithm, after each Type 1 phase, at least one edge is dropped. Furthermore, the very last phase must be of Type 1 and after this phase, an entire perfect matching, i.e., ${n \over 2}$ edges, are dropped. Therefore, the number of Type 1 phases is $m - {n \over 2} + 1$.

On the other hand, after each Type 2 phase, at least one more tight odd $\alpha$-cut is created and all previous tight odd $\alpha$-cuts are retained. We now show that a maximal laminar family of tight odd $\alpha$-cuts can contain at most ${n \over 2} - 1$ cuts. For this purpose, we will enhance the tree $\cT$ corresponding to $\Lc$, defined in Section \ref{sec.find-L}, as follows: For each $C \in (\Lc \cup \{V\})$, let $V_C \subset V$ be the set of vertices which are in $C$ but not in any $C' \in \Lc$ such that $C' \subset C$. For each $v \in V_C$, make $v$ a child of $C$. Repeating this for all $C \in (\Lc \cup \{V\})$ we get the enhanced tree $\cT'$.  

Next observe that if $C \in \Lc$ is a leaf of tree $\cT$, then $C$ must have at least three children in tree $\cT'$, all vertices. Furthermore, each internal node $c$ of $\cT$ has at least three children, which may be odd sets or singleton vertices. The following is a way of assigning at lest two distinct vertices to each $C \in \Lc$: For each leaf $C$ of $\cT$, of the three vertices which are children of $C$ in $\cT'$, assign any two to $C$ and carry the third to the parent of $C$. Each internal node $C$ of $\cT$, $C$ receives one vertex from each child which is an odd set and for each child that is a singleton, $C$ will use this vertex as follows: of these three vertex, two are assigned to $C$ and one is carried to the parent of $C$. Finally, observe that the root of $\cT'$, which is $V$, must have at least two children which may be odd sets or singleton vertices. Therefore, the root has two vertices which were not assigned to any $C \in \Lc$. 

Since each $C \in \Lc$ has been assigned two distinct vertices and two vertices are still left-over at the root, $|\Lc| \leq {n \over 2} -1$. This is also a bound of on the number of Type 2 phases. The lemma follows. 
\end{proof}

Since each phase of Algorithm \ref{alg.main} outputs one term, we get:

\begin{corollary}
	\label{cor.terms}
	The number of terms in the decomposition found by Algorithm \ref{alg.main} is at most $m$.
\end{corollary}

In placing an upper bound on the running time, special care is needed for bounding the growth of denominators, in the components of $y$, as a function of the number of phases of Algorithm \ref{alg.main}. 

\begin{lemma}
	\label{lem.largest}
	The largest denominator encountered by Algorithm \ref{alg.main} is bounded by $d^m \cdot n^{n/2}$, where $d$ and $m$ are defined in Notation \ref{not.edges}.
\end{lemma}

\begin{proof}
	We will include an additional step in Initialization in Algorithm \ref{alg.main}: Compute the lcm of all denominators occurring in the input vector $x$ and take that to be the denominator of each component of $x$; clearly, the lcm is bounded by $d^m$. Throughout the algorithm, we will ensure the following:
	
	{\bf Invariant:} All components of $y$ have the same denominator. 
	
As in Lemma \ref{lem.phases}, we say that a phase of Algorithm \ref{alg.main} is of {\em Type 1} if the condition in Step 2(f) is true and of {\em Type 2} otherwise. Lemma \ref{lem.phases} shows that the number of phases of Type 1 and 2 is upper bounded by $m$ and ${n \over 2}$, respectively. 

It is easy to check that because of the Invariant, a phase of Type 1 does not change the denominators. By Lemma \ref{lem.denominators}, if the denominators were $b$ at the beginning of a phase of Type 2, then the denominator of $\gamma$ computed is at most $b \cdot n$. The updated  vector $y$ and the updated $\alpha$ will also have this same bound. At this point, we will ensure that all components of the updated vector $y$ have the same denominator, hence ensuring the Invariant. 

Since the number of phases of Type 2 is bounded by ${n \over 2}$, the lemma follows. 
\end{proof}

\begin{theorem}
	\label{thm.main}
	A fractional perfect matching in a non-bipartite graph can be decomposed, in polynomial time, into a convex combination of perfect matchings with at most $m$ terms. In the worst case, Algorithm \ref{alg.main} requires $O(n^3 m^2)$ max-flow min-cut computations. 
\end{theorem}

\begin{proof}
	As a consequence of Lemma \ref{lem.largest}, binary search will take polynomial time. Therefore we have argued above that all steps of Algorithm \ref{alg.main} can be implemented in polynomial time and it finds a decomposition with at most $m$ terms.
	
	The running time is dominated by the time taken by Step 2(g)(v), i.e., the executions of the operation of update$(\Lc)$. As shown in Lemma \ref{lem.time-find-L}, each execution requires $O(nm^2)$ calls to the Padberg-Rao algorithm. Each call requires $O(n)$ max-flow min-cut computations. Finally, there are at most $n$ executions of update$(\Lc)$. Therefore, the algorithm requires $O(n^3 m^2)$ max-flow min-cut computations in the worst case.
\end{proof}

Next, consider a generalization in which the given non-bipartite graph $G = (V, E)$ has non-negative edge-weights given by $w$. Let $x$ be a minimum weight fractional perfect matching in $G$ and our task is to decompose it into a convex combination of minimum weight perfect matchings. Observe that if $w$ is ignored and $x$ is decomposed using Algorithm \ref{alg.main}, we will get the desired decomposition. This is so because the weight of $x$ equals the weight of the convex combination found. Therefore, if the latter has a perfect matching which is not of minimum weight, then it must have one which is lighter than $x$, giving a contradiction. Hence we get:

\begin{corollary}
	\label{cor.weighted}
	Given non-bipartite graph $G = (V, E)$ with non-negative edge-weights and a minimum weight fractional perfect matching $x$, Algorithm \ref{alg.main} decomposes $x$ into a convex combination of minimum weight perfect matchings. 
\end{corollary}

\begin{remark}
	\label{rem.bipartite}
	With the execution of successive phases, the support graph keeps getting sparser. As a result, it may get disconnected or be rendered bipartite. If the first possibility occurs, then the restriction of the current $x$ to each connected component is an $\alpha$-fractional perfect matching in that component. The second possibility is illustrated by graph $G_1$. For any choice of the three perfect matchings containing one edge from the set $\{(a, b), (c, d), (e, f)\}$ in the first phase, the remaining graph is bipartite and hence can be decomposed via the BvN Algorithm.
\end{remark}

	\section{Discussion}
\label{sec.discussion}

\cite{AVnari2019matching} observe the intriguing phenomenon that time and time again, algorithmic work on matching problems, from numerous perspectives, has proceeded in the following manner: first the bipartite case gets solved and much later, using substantially more mathematical and algorithmic machinery, the non-bipartite case gets solved. Moreover, the additional facts needed are different for different problems. Our extension of the Birkhoff-von Neumann Theorem to non-bipartite graphs is another example of this phenomenon. Is there a ``reason'' or explanation for this phenomenon?

We leave the open problem of determining if a tighter upper bound can be placed on the number of terms in our decomposition. Finally, it seems fruitful to look for efficient combinatorial algorithms for decomposing a point in the polyhedra of other combinatorial objects, such as the spanning tree polytope.

	\section{Acknowledgements}
\label{sec.ack}

I wish to thank Nima Anari, Jugal Garg and Martin Groetschel for valuable discussions, and Milena Mihail and Thorben Trobst for pointing out Corollary \ref{cor.weighted}.

	\bibliographystyle{alpha}
	\bibliography{refs}
\end{document}